\documentclass[aps,prl,groupedaddress,superscriptaddress,onecolumn,notitlepage,bibnotes,nofootinbib]{revtex4-1}

\usepackage{newpxtext,newpxmath}

\let\coloneqq\relax

\usepackage[latin1]{inputenc}
\usepackage{amsthm}
\usepackage{amssymb}
\usepackage{amsmath}
\usepackage{bbold}
\usepackage{bbm}
\usepackage[pdftex]{hyperref}
\usepackage{braket}
\usepackage{dsfont}
\usepackage{mathdots}
\usepackage{mathtools}
\usepackage{enumerate}
\usepackage[shortlabels]{enumitem}
\usepackage{csquotes}
\usepackage{stmaryrd}
\usepackage[cal=boondox]{mathalfa}
\usepackage{graphicx}
\usepackage{stackengine}
\usepackage{scalerel}
\usepackage{tensor}       
\usepackage{array}
\usepackage{makecell}
\newcolumntype{x}[1]{>{\centering\arraybackslash}p{#1}}
\usepackage{tikz}
\usepackage{pgfplots}
\usetikzlibrary{shapes.geometric, shapes.misc, positioning, arrows, arrows.meta, decorations.pathreplacing, decorations.pathmorphing, patterns, angles, quotes, calc}
\usepackage{booktabs}
\usepackage{xfrac}
\usepackage{siunitx}
\usepackage{centernot}
\usepackage{comment}
\usepackage{chngcntr}
\usepackage{caption}
\usepackage{subcaption}

\newtheorem{thm}{Theorem}
\newtheorem*{thm*}{Theorem}
\newtheorem{prop}[thm]{Proposition}
\newtheorem*{prop*}{Proposition}
\newtheorem{lemma}[thm]{Lemma}
\newtheorem*{lemma*}{Lemma}

\newtheorem*{cor*}{Corollary}

\newtheorem*{cj*}{Conjecture}

\newtheorem*{Def*}{Definition}

\newtheorem*{question*}{Question}

\newtheorem*{problem*}{Problem}

\makeatletter
\def\thmhead@plain#1#2#3{%
  \thmname{#1}\thmnumber{\@ifnotempty{#1}{ }\@upn{#2}}%
  \thmnote{ {\the\thm@notefont#3}}}
\let\thmhead\thmhead@plain
\makeatother

\theoremstyle{definition}

\newcommand{\bb}{\begin{equation}\begin{aligned}\hspace{0pt}}
\newcommand{\bbb}{\begin{equation*}\begin{aligned}}
\newcommand{\ee}{\end{aligned}\end{equation}}
\newcommand{\eee}{\end{aligned}\end{equation*}}
\newcommand*{\coloneqq}{\mathrel{\vcenter{\baselineskip0.5ex \lineskiplimit0pt \hbox{\scriptsize.}\hbox{\scriptsize.}}} =}

\newcommand{\eqt}[1]{\stackrel{\mathclap{\scriptsize \mbox{#1}}}{=}}
\newcommand{\leqt}[1]{\stackrel{\mathclap{\scriptsize \mbox{#1}}}{\leq}}

\newcommand{\ketbra}[1]{\ket{#1}\!\!\bra{#1}}

\newcommand{\e}{\varepsilon}
\renewcommand{\epsilon}{\varepsilon}
\newcommand{\G}{\mathrm{\scriptscriptstyle G}}

\newcommand{\id}{\mathds{1}}
\newcommand{\R}{\mathds{R}}

\DeclareMathOperator{\Tr}{Tr}

\DeclareMathAlphabet{\pazocal}{OMS}{zplm}{m}{n}

\DeclareMathOperator{\diag}{diag}
\DeclareMathOperator{\ran}{ran}

\newcommand{\MM}{\pazocal{M}}

\newcommand{\lsmatrix}{\left(\begin{smallmatrix}}
\newcommand{\rsmatrix}{\end{smallmatrix}\right)}

\newcommand{\rel}[3]{#1\big(#2\,\big\|\,#3\big)}

\stackMath

\stackMath

\makeatletter
\newcommand*\rel@kern[1]{\kern#1\dimexpr\macc@kerna}
\newcommand*\widebar[1]{%
  \begingroup
  \def\mathaccent##1##2{%
    \rel@kern{0.8}%
    \overline{\rel@kern{-0.8}\macc@nucleus\rel@kern{0.2}}%
    \rel@kern{-0.2}%
  }%
  \macc@depth\@ne
  \let\math@bgroup\@empty \let\math@egroup\macc@set@skewchar
  \mathsurround\z@ \frozen@everymath{\mathgroup\macc@group\relax}%
  \macc@set@skewchar\relax
  \let\mathaccentV\macc@nested@a
  \macc@nested@a\relax111{#1}%
  \endgroup
}

\counterwithin*{equation}{part}
\counterwithin*{thm}{part}
\counterwithin*{figure}{part}

\tikzset{meter/.append style={draw, inner sep=10, rectangle, font=\vphantom{A}, minimum width=30, line width=.8, path picture={\draw[black] ([shift={(.1,.3)}]path picture bounding box.south west) to[bend left=50] ([shift={(-.1,.3)}]path picture bounding box.south east);\draw[black,-latex] ([shift={(0,.1)}]path picture bounding box.south) -- ([shift={(.3,-.1)}]path picture bounding box.north);}}}
\tikzset{roundnode/.append style={circle, draw=black, fill=gray!20, thick, minimum size=10mm}}
\tikzset{squarenode/.style={rectangle, draw=black, fill=none, thick, minimum size=10mm}}

\definecolor{Blues5seq1}{RGB}{239,243,255}
\definecolor{Blues5seq2}{RGB}{189,215,231}
\definecolor{Blues5seq3}{RGB}{107,174,214}
\definecolor{Blues5seq4}{RGB}{49,130,189}
\definecolor{Blues5seq5}{RGB}{8,81,156}

\definecolor{Greens5seq1}{RGB}{237,248,233}
\definecolor{Greens5seq2}{RGB}{186,228,179}
\definecolor{Greens5seq3}{RGB}{116,196,118}
\definecolor{Greens5seq4}{RGB}{49,163,84}
\definecolor{Greens5seq5}{RGB}{0,109,44}

\definecolor{Reds5seq1}{RGB}{254,229,217}
\definecolor{Reds5seq2}{RGB}{252,174,145}
\definecolor{Reds5seq3}{RGB}{251,106,74}
\definecolor{Reds5seq4}{RGB}{222,45,38}
\definecolor{Reds5seq5}{RGB}{165,15,21}

\allowdisplaybreaks

\begin{document}

\title{A fully Gaussian quantum Stein's lemma}

\author{Filippo Girardi $^{\diamond}$}
\email{filippo.girardi@sns.it}
\affiliation{Scuola Normale Superiore, Piazza dei Cavalieri 7, 56126 Pisa, Italy}

\author{Jacopo Rizzo $^{\diamond}$}
\email{jacopo.rizzo@fu-berlin.de}
\affiliation{Dahlem Center for Complex Quantum Systems, Freie Universit\"at Berlin, 14195 Berlin, Germany}

\author{Leah Turner $^{\diamond}$}
\email{pmylt1@nottingham.ac.uk}
\affiliation{School of Mathematical Sciences and Centre for the Mathematics and Theoretical Physics of Quantum Non-Equilibrium Systems, University of Nottingham, University Park, Nottingham NG7 2RD, United Kingdom}

\author{Gerardo Adesso}
\email{gerardo.adesso@nottingham.ac.uk}
\affiliation{School of Mathematical Sciences and Centre for the Mathematics and Theoretical Physics of Quantum Non-Equilibrium Systems, University of Nottingham, University Park, Nottingham NG7 2RD, United Kingdom}

\author{Ludovico Lami}
\email{ludovico.lami@sns.it}
\affiliation{Scuola Normale Superiore, Piazza dei Cavalieri 7, 56126 Pisa, Italy}

\begin{abstract}
A central yet mysterious problem in quantum information is the asymptotic distinguishability of quantum states under restricted measurements. Here we look at asymptotic asymmetric hypothesis testing for general multi-mode bosonic Gaussian states restricted to arbitrary collective Gaussian measurements. We establish a single-letter formula for the Stein exponent for many pairs of states, featuring two distinct contributions, corresponding to first and second moments, respectively. Remarkably, the Gaussian-measured relative entropy exhibits genuine non-additivity, but the only effect of additivity violations is to asymptotically decouple first and second moments. We prove this by establishing sufficient conditions for additivity of the second-moment term via novel log-determinant matrix inequalities. Our results not only identify a single-letter expression for the Gaussian Stein exponent in a wide range of practical situations, but also rule out any possible advantage from collective Gaussian measurements when considering centred Gaussian states.
\end{abstract}

\maketitle

\begingroup
\renewcommand{\thefootnote}{$\diamond$}
\footnotetext{These authors contributed equally to this work.}
\endgroup

\section{Introduction}
Gaussian states, transformations, and measurements lie at the heart of many implementations of photonic technologies, both due to their elegant description in phase space and ease of implementation experimentally \cite{weedbrook12,adesso14,BUCCO}. On the other hand, restricting to only Gaussian operations can impose limitations on state discrimination tasks, even when only Gaussian states are considered \cite{Takeoka-Sasaki,hnhp-jhr2,PhysRevA.103.022423}. This makes the fully Gaussian setting a natural framework in which to understand how constraints on the available measurements affect fundamental information processing tasks. In this paper, we study quantum hypothesis testing when both hypotheses are multimode Gaussian states and the allowed measurements are restricted to be Gaussian. 
Such an investigation contributes to and complements the broader study of fully Gaussian quantum-information protocols \cite{G-resource-theories,PhysRevA.71.032336,turner2026optimaldiscriminationgaussianstates,Oh_2019,walsh2025allgaussianstatediscriminationcoherent} and of quantum state discrimination under restricted classes of measurements \cite{Piani2009,VV-dh,brandao_adversarial,11007117,Matthews2010,Y_ng_ez_2026,Leone_2025,VV-dh-Chernoff}.


In the standard i.i.d.\ quantum hypothesis testing setting, one is given $n$ identical copies of a quantum system, which are promised to be prepared according to one of two possible hypotheses: either the state is $\rho^{\otimes n}$, corresponding to the null hypothesis, or it is $\sigma^{\otimes n}$, corresponding to the alternative hypothesis. The task is to distinguish between these two possibilities by performing a measurement on the joint $n$-copy system and, from the observed outcome, decide which hypothesis is more likely to be correct. Any such decision procedure is characterised by a pair of error probabilities, namely the probability of rejecting the null hypothesis when it is true and the probability of accepting it when the alternative hypothesis is true. These are called type-I and type-II errors, respectively, and their probabilities generally obey a trade-off. A fundamental setting is given by asymmetric hypothesis testing: here one minimises the type-II error subject to a bound on the type-I error. 
In the asymptotic limit, the \emph{Stein exponent} fully characterises the optimal exponential decay rate of the type-II error under a fixed type-I error bound $\varepsilon\in(0,1)$. Here we study Gaussian-restricted hypothesis testing in the regime where the type-I error bound is taken to zero after the asymptotic limit. With unrestricted collective measurements, the Stein exponent coincides with the quantum relative entropy \cite{Hiai1991,Ogawa2000}. 
We study the fully Gaussian version of this problem. 
Specifically, we consider two Gaussian hypotheses $\rho_\G$ and $\sigma_\G$, allowing arbitrary collective Gaussian measurements \cite{weedbrook12,adesso14,BUCCO}. Whereas the unrestricted Stein exponent admits a single-letter characterisation in terms of the quantum relative entropy, its Gaussian counterpart is given by the \emph{regularised Gaussian-measured relative entropy}. This brings the associated \emph{additivity problem} to the forefront: can collective Gaussian measurements outperform independent single-copy measurements, and what mechanism enables such an advantage? More generally, under which conditions can the regularisation be evaluated to obtain a single-letter characterisation of the Gaussian Stein exponent?

Here, we make substantial progress towards this goal through two complementary results. First, we establish an exact asymptotic separation of the contributions of first and second moments. By applying a mixing transformation across the copies, the entire displacement difference can be concentrated into a single output copy without altering the covariance matrices. The first-moment contribution can thus be optimised on that copy, while the remaining copies are devoted to covariance discrimination. Since the loss of one copy has no asymptotic effect, this mechanism yields an explicit displacement contribution for arbitrary Gaussian hypotheses and confines the remaining additivity problem to their covariances only.
Second, we establish a general structural criterion for exact additivity of the centred Gaussian-measured relative entropy. More precisely, we require the reference covariance matrix $W$ to have a flat symplectic spectrum, so that $SWS^\intercal=w\id_{2m}$ for some symplectic $S$, and the transformed covariance $\widetilde V=SVS^\intercal$ to satisfy $[\widetilde V,\Omega\widetilde V\Omega^\intercal]=0$. We show that these input conditions alone rule out any advantage from collective Gaussian measurements on the centred states, for every number of copies $n$. The criterion encompasses all single-mode pairs and all pure multimode pairs, as well as families of mixed multimode states. Combined with the separation of moments, it yields single-letter characterisations of the Gaussian Stein exponent throughout this class.

\section{Summary of main results}
In a Gaussian hypothesis-testing problem, we consider $n$ copies of an $m$-mode continuous-variable quantum system, together with two hypotheses for the state in which it has been prepared: $\rho_\G^{\otimes n}$, corresponding to the null hypothesis, and $\sigma_\G^{\otimes n}$, corresponding to the alternative hypothesis. Here, $\rho_\G$ and $\sigma_\G$ are two given Gaussian states, described by quantum covariance matrices (QCMs) $V,W$ and first moments $s,t$, respectively. Our goal is to decide which of the two hypotheses is correct. To this end, we can choose an arbitrary Gaussian measurement on the full $nm$-mode system, perform it, and use classical post-processing of the measurement outcome to formulate our decision. We refer to any such procedure as a \emph{Gaussian strategy}.
For any given Gaussian strategy, there is an associated pair of error probabilities $(\alpha_n,\beta_n)$. Here, $\alpha_n$, called the \emph{type-I error probability}, is the probability of incorrectly accepting the alternative hypothesis $\sigma_\G^{\otimes n}$ when the null hypothesis $\rho_\G^{\otimes n}$ is actually correct, and vice versa for the \emph{type-II error probability} $\beta_n$. We denote by $\beta_{n,\e}\big(\rho_\G^{\otimes n}\,\big\|\,\sigma_\G^{\otimes n}\big)$ the smallest type-II error probability achievable by a Gaussian strategy while keeping $\alpha_n\leq\e$. Clearly, there is a trade-off between $\alpha_n$ and $\beta_n$: requiring the former to be smaller generally forces the latter to be larger. As usual, we are interested in the asymptotic rate of decay of the type-II error probability for a fixed, small type-I error probability. We then take $\varepsilon\to0^+$. This quantity characterises the optimal performance of hypothesis testing in the fully Gaussian setting, in which both the states and the allowed measurements are Gaussian, and we call it the \emph{Gaussian Stein exponent}. Formally, it is given by
\bb
\mathrm{Stein}_\G(\rho_\G\|\sigma_\G) \coloneqq \lim_{\e\to 0^+} \liminf_{n\to\infty} \left\{ - \frac1n \log \beta_{n,\e}\big(\rho_\G^{\otimes n}\,\big\|\,\sigma_\G^{\otimes n}\big) \right\} .
\ee
If \emph{arbitrary} quantum measurements are allowed, the asymptotic error exponent is given by the quantum relative entropy, according to the quantum Stein's lemma~\cite{Hiai1991,Ogawa2000}, and can be evaluated explicitly for Gaussian states~\cite{LL-Renyi}. Here we take a different route: rather than lifting the Gaussian restriction, we retain it throughout and ask how much asymptotic distinguishability remains accessible to Gaussian measurements alone. This follows the general viewpoint of quantum state discrimination under restricted classes of measurements~\cite{Piani2009,VV-dh,brandao_adversarial}, but with the restriction imposed by Gaussian quantum mechanics.

As a first step, we show that this asymptotic discrimination problem is exactly characterised by the regularisation of the corresponding restricted measured relative entropy. For any two quantum states $\omega,\tau$, we define the \emph{Gaussian-measured relative entropy}
\begin{align}
D^\G(\omega\|\tau)
\coloneqq
\sup_{\MM_\G}
D\big(
\MM_\G(\omega)
\big\|
\MM_\G(\tau)
\big),
\label{Gaussian_Stein_simplified_2}
\end{align}
where $D$ denotes the \emph{classical} relative entropy between the probability distributions obtained after the measurement, and the supremum is taken over all Gaussian measurements $\MM_\G$. We then obtain the following characterisation of the Gaussian Stein exponent (see Proposition~\ref{prop:regstein} in the Appendix):
\begin{align}
\mathrm{Stein}_\G(\rho_\G\|\sigma_\G)
=
\lim_{n\to\infty}
\frac1n
D^\G\big(
\rho_\G^{\otimes n}
\big\|
\sigma_\G^{\otimes n}
\big).
\label{Gaussian_Stein_simplified_11}
\end{align}
The main difficulty is therefore the regularisation in Eq.~\eqref{Gaussian_Stein_simplified_11}. At the $n$-copy level, the optimisation ranges over arbitrary collective Gaussian measurements acting jointly across all copies, and it is not a priori clear whether such measurements can outperform independent single-copy measurements. The evaluation of the Gaussian Stein exponent thus becomes an additivity problem for the Gaussian-measured relative entropy. In other words, the regularisation probes whether the Gaussian restriction changes only the attainable distinguishability rate, or also introduces a genuinely many-copy structure to the optimal measurement strategy. If additivity holds, the asymptotic exponent immediately admits a single-letter characterisation; conversely, any failure of additivity would reveal a genuine advantage of collective Gaussian measurements. Determining whether, and under which conditions, this regularisation can be removed is therefore a central question in the characterisation of asymptotic state discrimination under Gaussian measurement restrictions.

This brings us to the central question of this work:
\begin{center}
\emph{Can the Gaussian Stein exponent be evaluated exactly for arbitrary Gaussian states? Equivalently, can the regularisation in Eq.~\eqref{Gaussian_Stein_simplified_11} be solved, and when can it be removed altogether?}
\end{center}
In the next sections, we make substantial progress on resolving this open problem.

\subsection{Decoupling of first and second moments}
As a first step, we show that the Gaussian Stein exponent admits a sharp structural decomposition: the entire dependence on the first moments $s,t$ of the two Gaussian states can be solved exactly and is intrinsically single-letter. As a consequence, the full difficulty of the regularisation problem is carried by the covariance matrices alone.

\begin{thm}[(Asymptotic separation of moments (Theorem \ref{prop:momentsepapp} in the Appendix))]\label{thm:momentsep}
For any pair of Gaussian states $\rho_\G$ and $\sigma_\G$, with first moments $s$ and $t$, and covariance matrices $V$ and $W$, respectively, the corresponding Gaussian Stein exponent factorises as
\bb\label{eq:sfactor}
\mathrm{Stein}_\G(\rho_\G\|\sigma_\G) = \mathrm{Stein}_\G(\rho^0_\G\|\sigma^0_\G) + (s-t)^\intercal W^{-1}(s-t), 
\ee
where $\rho^0_\G$ and $\sigma^0_\G$ are the Gaussian states obtained by setting the first moments of $\rho_\G$ and $\sigma_\G$ to zero.
\end{thm}

Theorem~\ref{thm:momentsep} reveals a fundamental asymmetry between first and second moments in the Gaussian Stein problem. The displacement information is intrinsically single-letter and can be optimised independently, so any nontrivial regularisation must originate from the covariance matrices alone. Indeed, a Gaussian measurement with seed QCM $\gamma$ maps a Gaussian state with QCM $V$ and first moments $s$ to a classical Gaussian distribution with mean $s$ and covariance matrix $(V+\gamma)/2$. The classical relative entropy between the two resulting distributions therefore decomposes into a contribution depending only on the covariance matrices and a quadratic term depending on the displacement difference $s-t$. In the asymptotic Stein setting, the result above shows that this displacement contribution can be optimised independently and is exactly single-letter. After removing this explicitly solvable term and inserting the closed formula for the relative entropy of classical Gaussian distributions, the Gaussian Stein exponent reduces to the centred covariance-matrix optimisation

\bb
\begin{aligned}
\mathrm{Stein}_\G(\rho^0_\G\|\sigma^0_\G)
=
\lim_{n\to\infty}\frac1n
\sup_{\gamma\geq i\Omega^{\oplus n}}
\bigg\{
&\frac12
\log
\frac{\det\!\left(W^{\oplus n}+\gamma\right)}
{\det\!\left(V^{\oplus n}+\gamma\right)}
+
\frac12
\Tr\!\left[
\left(W^{\oplus n}+\gamma\right)^{-1}
(V-W)^{\oplus n}
\right]
\bigg\}.
\end{aligned}
\label{eq:Stein_second_moments}
\ee
Thus, from this point onward, the additivity and regularisation questions concern only the covariance-matrix optimisation in Eq.~\eqref{eq:Stein_second_moments}. This isolates the genuinely many-copy part of the problem: the displacement sector is asymptotically solved, whereas any further possible collective Gaussian advantage must arise from how the measurements resolve differences in noise, squeezing, and correlations encoded in the covariance matrices.

\subsection{The additivity problem}
The covariance contribution remains a regularised quantity and is thus intractable to calculate. We therefore wish to study whether the term to be optimised,
\bb
\frac12
\log
\frac{\det\!\left(W^{\oplus n}+\gamma\right)}
{\det\!\left(V^{\oplus n}+\gamma\right)}
+
\frac12
\Tr\!\left[
\left(W^{\oplus n}+\gamma\right)^{-1}
(V-W)^{\oplus n}
\right],
\ee
is additive over the number of copies $n$, so that we can reduce the optimisation to measurements acting on a single copy. We find sufficient conditions under which this term does become additive, outlined in the following Theorem. 

\begin{thm}[(Sufficient conditions for additivity)]\label{thm:sufficientmain} Let \(V,W\) be \(m\)-mode quantum covariance matrices. Assume that \(W\) has flat symplectic spectrum, namely, there is a symplectic matrix \(S\) such that $SWS^\intercal = w\id_{2m}$.
If, in addition, 
\bb\label{eq:condition}
[\widetilde V,\Omega\widetilde V\Omega^\intercal]=0,
\ee
where  $\widetilde V\coloneqq SVS^\intercal$, then we have
\bb
    \rel{D^{\G}}
{V^{\oplus n}}{W^{\oplus n}}
=
n\rel{D^{\G}}{V}{W}.
\ee

\end{thm}
Theorem~\ref{thm:sufficientmain} gives a concrete structural criterion for when collective Gaussian measurements are unnecessary. The condition that $W$ have a flat symplectic spectrum means that, after a suitable Gaussian change of coordinates, the reference covariance matrix is proportional to the identity. The additional commutation condition
\[
[\widetilde V,\Omega\widetilde V\Omega^\intercal]=0
\]
ensures that the covariance structure of the other hypothesis can be brought, by passive Gaussian transformations, to a form in which the relevant quadratures decouple mode by mode. In this situation, collective mode mixing cannot improve the asymptotic discrimination rate beyond what is already achievable with optimised single-copy Gaussian measurements.

The criterion covers, in particular, the important cases of all pairs of general mixed single-mode Gaussian states and all pairs of pure multimode Gaussian states, as well as further families of mixed multimode states. As such, for these states, and more generally for any pair satisfying Theorem~\ref{thm:sufficientmain}, combining this result with Theorem~\ref{thm:momentsep}, we see the Gaussian Stein exponent in Eq.~\eqref{Gaussian_Stein_simplified_11} drastically simplifies to the single-letter optimisation
\bb
\mathrm{Stein}_\G(\rho_\G\|\sigma_\G) =
\sup_{\gamma\geq i\Omega}
\bigg\{
&\frac12
\log
\frac{\det\!\left(W+\gamma\right)}
{\det\!\left(V+\gamma\right)}
+
\frac12
\Tr\!\left[
\left(W+\gamma\right)^{-1}
(V-W)
\right]
\bigg\} + (s-t)^\intercal W^{-1}(s-t).
\ee
The significance of this reduction is that the asymptotic problem no longer contains any irreducible many-copy optimisation: all remaining difficulty is finite-dimensional and confined to the optimisation of a single-copy Gaussian measurement.

\section{Discussion}

We have shown that, in the all-Gaussian hypothesis-testing scenario, the key quantity governing the asymptotic decay rate of the type-II error is the regularised Gaussian-measured relative entropy, or Gaussian Stein exponent. In the asymptotic limit, this quantity separates into two contributions: one depending on the first moments of the states, and one depending only on their covariance matrices. The first-moment contribution is intrinsically single-letter, admits a closed-form solution, and can be optimised independently of the covariance contribution. Consequently, any nontrivial regularisation, and any possible advantage arising from genuinely collective Gaussian measurements, is entirely associated with the second moments.

We have also identified sufficient conditions under which this covariance-only contribution is additive in the number of copies and therefore admits a single-letter characterisation. These conditions include, in particular, all pairs of single-mode Gaussian states and all pairs of pure multimode Gaussian states. Whenever the conditions of Theorem~\ref{thm:sufficientmain} are satisfied, the regularisation can be removed and the Gaussian Stein exponent reduces to an optimisation over a single-copy Gaussian measurement, together with the explicit displacement contribution derived in Theorem~\ref{thm:momentsep}. For centred states belonging to this class, collective Gaussian measurements therefore provide no asymptotic advantage over independent single-copy measurements.

Our results also have a direct operational interpretation. The contribution of the first moments can be isolated asymptotically by a passive Gaussian mixing of the copies that concentrates the displacement difference into a single output copy, leaving the covariance matrices unchanged. The remaining copies can then be used to discriminate the covariance matrices. Thus, the asymptotic treatment of first and second moments can be separated physically as well as mathematically. In the additive cases identified here, the covariance contribution can subsequently be extracted using only single-copy Gaussian measurements, avoiding the need for genuinely collective measurements across many copies.

The most natural direction for future work is to determine necessary and sufficient conditions for additivity of the covariance-only contribution, extending Theorem~\ref{thm:sufficientmain}. It remains open whether there exist pairs of Gaussian states for which this contribution is genuinely non-additive, and hence whether collective Gaussian measurements can provide a strict asymptotic advantage for centred Gaussian-state discrimination. A complete solution of this question would yield a full characterisation of the Gaussian Stein exponent for arbitrary multimode Gaussian states and clarify the precise role of collective Gaussian measurements in asymptotic hypothesis testing.

\begin{acknowledgments}
\section{Acknowledgments}
LL and FG acknowledge financial support from the European Union under the European Research Council (ERC Grant Agreement No.~101165230).
LT and GA acknowledge financial support from the Engineering and Physical Sciences Research Council (EPSRC Grant No.~EP/W524402/1). J.R. was supported by the BMFTR (QR.N), the Clusters of Excellence (ML4Q, MATH+), QuantERA (SDPCode), the Munich Quantum Valley, Berlin Quantum, the Quantum Flagship (Millenion, PASQuanS2), the DFG (CRC 183), and the European Research Council (DebuQC).
We acknowledge early discussions with Bartosz Regula on the topic of this work.

\subsection{AI Statement}
The initial ideas, special cases, examples, interpretations, and core analytical approach underlying this work were developed by the authors. Most main results were established at this stage. At later stages, large language models, including ChatGPT (GPT-5.5, GPT-5.6 Sol, and GPT-6 Astra), were used as research assistants to help explore extensions of the arguments, refine proofs of some of the general results, and improve parts of the presentation. All mathematical statements, proofs, and conclusions included in the manuscript were independently checked and verified by the authors. \end{acknowledgments}

\bibliographystyle{apsrevfixedwithtitles}
\bibliography{biblio}

\bigskip
\appendix

\newpage
\section*{APPENDIX}
\section{Preliminaries on Gaussian states and measurements}
\label{sec:gaussian_preliminaries}

We briefly recall the phase-space formalism for Gaussian quantum systems; see, e.g.,~\cite{BUCCO} for a comprehensive introduction. 
An $m$-mode bosonic system is described by a vector of canonical quadrature operators
\bb
R\coloneqq (Q_1,P_1,\ldots,Q_m,P_m)^\intercal,
\ee
satisfying the canonical commutation relations
\bb
[R_j,R_k]=i\Omega_{jk},
\qquad
\Omega\coloneqq \bigoplus_{j=1}^m \Omega_1,
\qquad
\Omega_1\coloneqq
\begin{pmatrix}
0&1\\
-1&0
\end{pmatrix}.
\ee
Equivalently, we may define the symplectic form
\bb
\Omega\coloneqq\begin{pmatrix}
        0&\mathds{1}\\-\mathds{1}&0
        \end{pmatrix},
\ee
which corresponds to the quadrature ordering $R\coloneqq (Q_1,\ldots,Q_m,P_1,\ldots P_m)^\intercal$.
A Gaussian state $\rho_\G$ is completely specified by its vector of first moments
\bb
s_j\coloneqq \Tr[\rho_\G R_j]
\ee
and its quantum covariance matrix (QCM)
\bb
V_{jk}
\coloneqq
\Tr\!\left[
\rho_\G
\left\{
R_j-s_j,R_k-s_k
\right\}
\right].
\ee
The uncertainty principle is equivalent to the condition
\bb
V\geq i\Omega.
\ee
Conversely, every real symmetric matrix satisfying this condition is the QCM of a Gaussian state. 
Gaussian unitaries act affinely on phase space, transforming
\bb
R\longmapsto SR+d,
\ee
where $d\in\R^{2m}$ and $S$ is symplectic,
\bb
S\Omega S^\intercal=\Omega.
\ee
Accordingly, the first and second moments transform as
\bb
s\longmapsto Ss+d,
\qquad
V\longmapsto SVS^\intercal.
\ee
By Williamson's Theorem, every QCM can be written as
\bb
V
=
S
\left(
\bigoplus_{j=1}^m \nu_j\id_2
\right)
S^\intercal,
\label{eq:will}
\ee
where $S$ is symplectic and $\nu_j\geq1$ are the symplectic eigenvalues of $V$. In particular, a Gaussian state is pure if and only if $\nu_j=1$ for all $j$.
A Gaussian measurement can be parameterised by a Gaussian seed with QCM
\bb
\gamma\geq i\Omega.
\ee
When a Gaussian measurement with seed $\gamma$ is performed on a Gaussian state with QCM $V$ and first moments $s$, its outcome is a classical Gaussian random variable with mean $s$ and covariance matrix
\bb
\frac{V+\gamma}{2}.
\ee
We denote the corresponding measurement by $\MM_\gamma$. Homodyne measurements are recovered as limiting cases in which the seed becomes infinitely squeezed. Tensor products of Gaussian measurements are again Gaussian measurements, with
\bb
\MM_{\gamma_1}\otimes\MM_{\gamma_2}
=
\MM_{\gamma_1\oplus\gamma_2}.
\ee
We will also repeatedly use the classical relative entropy between two Gaussian probability distributions with covariance matrices $A,B>0$ and first moments $a,b$. In our notation,
\bb
D(A,a\|B,b)
=
\frac12\log\frac{\det B}{\det A}
+
\frac12
\Tr\!\left[
B^{-1}
\left(
A-B+(a-b)(a-b)^\intercal
\right)
\right].
\label{Classical_RelEnt}
\ee
Consequently, for two Gaussian states $\rho_\G$ and $\sigma_\G$ with QCMs $V,W$ and first moments $s,t$, respectively,
\bb
D\!\left(
\MM_\gamma(\rho_\G)
\,\big\|\,
\MM_\gamma(\sigma_\G)
\right)
=
D\!\left(
\frac{V+\gamma}{2},s
\,\bigg\|\,
\frac{W+\gamma}{2},t
\right).
\label{eq:gaussmeasD}
\ee

\section{Hypothesis testing\\with Gaussian states and measurements}
We now establish the reduction of the Gaussian Stein exponent to the regularised Gaussian-measured relative entropy. Recall that $\beta_{n,\e}$ denotes the optimal type-II error achievable by a Gaussian strategy under the constraint $\alpha_n\leq\e$. We will show that the corresponding asymptotic exponent is exactly given by the regularisation of $D^\G$.

\begin{prop}[(Regularised Gaussian Stein's lemma)]\label{prop:regstein} We have
\begin{align}
\mathrm{Stein}_\G(\rho_\G\|\sigma_\G) =&\ \lim_{n\to \infty} \frac1n\, D^\G\big(\rho_\G^{\otimes n}\,\big\|\,\sigma_\G^{\otimes n}\big).
\label{eq:steinreg}
\end{align}
\begin{proof}
We first show that the limit in~\eqref{eq:steinreg} exists. This follows by a standard argument, noting that the following super-additivity inequality holds for the measured Gaussian relative entropy
\bb
\begin{aligned}
D^\G\big(\rho_\G^{\otimes (n+m)}\,\big\|\,\sigma_\G^{\otimes (n+m)}\big)
&\geq
D\left(
(\MM_n\otimes\MM_m)(\rho_\G^{\otimes(n+m)})
\,\big\|\,
(\MM_n\otimes\MM_m)(\sigma_\G^{\otimes(n+m)})
\right)
\\
&=
D\left(
\MM_n(\rho_\G^{\otimes n})
\,\big\|\,
\MM_n(\sigma_\G^{\otimes n})
\right)
+
D\left(
\MM_m(\rho_\G^{\otimes m})
\,\big\|\,
\MM_m(\sigma_\G^{\otimes m})
\right) \\
&=
D^\G\left(
\rho_\G^{\otimes n}
\,\big\|\,
\sigma_\G^{\otimes n}
\right)
+
D^\G\left(
\rho_\G^{\otimes m}
\,\big\|\,
\sigma_\G^{\otimes m}
\right)
\end{aligned}
\ee
where in the first line we used that tensor products of Gaussian measurements are valid Gaussian measurements, in the second line we used the additivity of the classical relative entropy for factorised probability distributions, and in the last line we have taken the supremum over the Gaussian measurements $\MM_n,\MM_m$. Now, by Fekete's lemma, the limit exists and is equal to the supremum
\bb
\lim_{n\to \infty} \frac1n\, D^\G\big(\rho_\G^{\otimes n}\,\big\|\,\sigma_\G^{\otimes n}\big) = \sup_{n\geq 1} \frac1n\, D^\G\big(\rho_\G^{\otimes n}\,\big\|\,\sigma_\G^{\otimes n}\big).
\ee
Let us call this limit $D^\G_\infty\big( \rho_G\,\big\|\,\sigma_\G\big)$. We now prove that the Stein exponent is a lower bound. Consider any Gaussian strategy on $n$ copies with type-I error $\alpha_n$ and type-II error $\beta_n$. Write its Gaussian measurement $\MM_n$, then the final hypothesis testing decision is a classical channel $\big(\MM_n(\rho_\G^{\otimes n}),\MM_n(\sigma_\G^{\otimes n})\big) \to \big( (1-\alpha_n,\alpha_n), (\beta_n,1-\beta_n) \big)$. Then, by data-processing of the classical relative entropy
\bb
D^\G\big(\rho_\G^{\otimes n}\,\big\|\,\sigma_\G^{\otimes n}\big) &\geq D\left(
\MM_n(\rho_\G^{\otimes n})
\,\big\|\,
\MM_n(\sigma_\G^{\otimes n})
\right) \\
&\geq D\left( (1-\alpha_n,\alpha_n) \,\big\|\, (\beta_n,1-\beta_n) \right) \\
&= (1-\alpha_n) \log\frac{1-\alpha_n}{\beta_n} + \alpha_n \log\frac{\alpha_n}{1-\beta_n} \\
&=-h_2(\alpha_n) -(1-\alpha_n)\log\beta_n -\alpha_n\log(1-\beta_n)
\\
&\geq-h_2(\alpha_n)-(1-\alpha_n)\log\beta_n,
\ee
where we have defined $h_2(x)\coloneqq -x\log x-(1-x)\log (1-x)$.
Whenever $\alpha_n\leq \varepsilon$, this establishes
\bb
-\frac1n\log\beta_n \leq \frac{D^\G\big(\rho_\G^{\otimes n}\,\big\|\,\sigma_\G^{\otimes n}\big)+\log 2}{n(1-\varepsilon)}.
\ee
In particular, this holds for the optimal $\beta_{n,\varepsilon}$, therefore
\bb
\liminf_{n\to\infty}
-\frac1n\log\beta_{n,\varepsilon} \leq \frac{D^\G_\infty\big( \rho_G\,\big\|\,\sigma_\G\big)}{1-\varepsilon}.
\ee
Letting $\varepsilon \to 0^+$ gives
\bb
\mathrm{Stein}_\G(\rho_\G\|\sigma_\G) \leq D^\G_\infty\big( \rho_G\,\big\|\,\sigma_\G\big).
\ee
We now prove the converse upper bound. By the definition of measured Gaussian relative entropy, for any $\delta>0$ we can choose a Gaussian measurement $\MM_k$ on $k$ input copies satisfying
\bb
D\left( \MM_k(\rho_\G^{\otimes k})\,\big\|\,\MM_k(\sigma_\G^{\otimes k})\right) \geq D^\G\big(\rho_\G^{\otimes k}\,\big\|\,\sigma_\G^{\otimes k}\big) - \delta.
\label{eq:measgauss}
\ee
Now split $n = rk + s$ for some integers $r,s$ with $s<k$. On the first $rk$ copies we perform the Gaussian measurement $\MM_k^{\otimes r}$ and we ignore the remaining $s$ copies. Since the remaining problem now is about distinguishing classical probability distributions, we now invoke the classical Stein's lemma. Let $0<\varepsilon<1$, then there exists a classical decision rule with type-I error $\alpha_r\leq \varepsilon$ such that the type-II $\{\beta_r\}_{r=1}^\infty$ error sequence satisfies
\bb
\lim_{r\to\infty} -\frac1r\log\beta_r = D\left( \MM_k(\rho_\G^{\otimes k})\,\big\|\,\MM_k(\sigma_\G^{\otimes k})\right).
\label{eq:gaussblock}
\ee
We now look at the full $\MM_k^{\otimes r}$ strategy, when $n\to\infty$, the resulting type-II error on the $n$ copies satisfies 
\bb
\liminf_{r\to\infty} \frac{r}{rk+s}\left(-\frac1r\log\beta_r\right)
&= \frac{1}{k} D\left( \MM_k(\rho_\G^{\otimes k})\,\big\|\,\MM_k(\sigma_\G^{\otimes k})\right) \\
&\geq \frac{1}{k} \left( D^\G\big(\rho_\G^{\otimes k}\,\big\|\,\sigma_\G^{\otimes k}\big) - \delta\right).
\ee
Where in the first equality we used \eqref{eq:gaussblock} and in the second line we used \eqref{eq:measgauss}.
Sending $k\to\infty$ concludes the proof of the achievability and gives \eqref{eq:steinreg}.
\end{proof}
\end{prop}



This first result will be fundamental in what follows, since it allows us to turn the Gaussian hypothesis testing problem into a classical surrogate. In particular, 
the output of a Gaussian measurement with a seed $\gamma$ on a Gaussian state with QCM $V$ and first moments $s$ is a Gaussian probability distribution with covariance matrix $\frac{V+\gamma}{2}$ and first moments $s$~\cite[5.139]{BUCCO}. Because of this, using~\eqref{Classical_RelEnt}, we can rephrase the result in proposition~\ref{prop:regstein} as follows:
\bb
\mathrm{Stein}_\G(\rho_\G\|\sigma_\G) &= \lim_{n\to\infty} \frac1n\, \sup_{\gamma \geq i\Omega^{\oplus n}} D\bigg( \frac{V^{\oplus n} + \gamma}{2},\, s^{\oplus n} \,\bigg\|\, \frac{W^{\oplus n} + \gamma}{2},\, t^{\oplus n} \bigg) \\
&= \lim_{n\to\infty} \frac1n\, \sup_{\gamma \geq i\Omega^{\oplus n}} D\Big( V^{\oplus n} + \gamma,\, \sqrt{2}\, (s-t)^{\oplus n} \,\Big\|\, W^{\oplus n} + \gamma,\, 0 \Big) \\
&= \lim_{n\to\infty} \frac1n \sup_{\gamma \geq i\Omega^{\oplus n}} \Bigg\{ \frac12 \log \frac{\det \left(W^{\oplus n} \!+\! \gamma \right)}{\det \left(V^{\oplus n} \!+\! \gamma\right)} \\
&\hspace{16.7ex} + \frac12 \Tr \left[ \left(W^{\oplus n} \!+\! \gamma \right)^{-1} \left( (V\!-\!W)^{\oplus n} + 2 \left( (s\!-\!t)^{\oplus n} \right) \left( (s\!-\!t)^{\oplus n} \right)^\intercal \right) \right] \Bigg\}\, .
\label{Stein_calculation_0}
\ee
Note that the quantity inside the optimisation is monotonically non-increasing in $\gamma$, and thus we can restrict the optimisation to pure measurements only.
The central question is whether the quantity inside the optimisation is additive under tensor products. If this were the case, the regularisation over $n$ copies would collapse to the single-copy problem, and it would suffice to optimise over a single-system Gaussian seed QCM $\gamma$. Establishing such an additivity property would therefore already provide a \emph{single-letter} characterisation of the Gaussian Stein exponent, even in the absence of a closed-form solution to the remaining finite-dimensional optimisation.

\section{The first moments term}
We begin by focusing on the simplest case: that of two Gaussian states with the same covariance matrices but differing first moments, such as two coherent states. Here we are only left with one term in the Stein exponent \eqref{Stein_calculation_0}, namely we aim to calculate
\bb
\mathrm{Stein}_\G(\rho_\G\|\sigma_\G)
&= \lim_{n\to\infty} \frac1n \sup_{\gamma \geq i\Omega^{\oplus n}} \Bigg\{ \frac12 \Tr \left[ \left(V^{\oplus n} \!+\! \gamma \right)^{-1} \left(2 \left( (s\!-\!t)^{\oplus n} \right) \left( (s\!-\!t)^{\oplus n} \right)^\intercal \right) \right] \Bigg\}\, .
\ee
In this case, we find the problem becomes additive over the number of copies of the states, and we recover a single-letter, closed form expression.

\begin{prop}\label{thm:firstmoments}
    Let $\rho_G$ and $\sigma_G$ be Gaussian states with first moments $s$ and $t$ respectively, and both with second moments $V$. Then $D^G(\rho_G^{\otimes n}(V,s)\|\sigma_G^{\otimes n}(V,t))$ is additive over $n$. Furthermore, it is achieved by a homodyne measurement.
\end{prop}
\begin{proof}
Start by considering two single-mode Gaussian states $\rho_\G$ and $\sigma_\G$ with the same QCM and different displacements, say
$$V=W=\left(\begin{array}{cc}a & c \\ c & b\end{array}\right)\quad \text{and}\quad s-t=\left(\begin{array}{c} x \\ y\end{array}\right).$$ 
Write a general pure single-mode Gaussian measurement as
\begin{equation}\label{pureseed}
    \gamma=R_\theta \cdot {\rm diag}(z, 1/z) \cdot R_\theta^\intercal,
\end{equation}
with $z >0$ and $R_\theta = \left(
\begin{array}{cc}
 \cos \theta  & -\sin \theta  \\
 \sin \theta  & \cos \theta  \\
\end{array}
\right)$. The optimal measurement parameter values are $z \rightarrow 0$ and $\tan \theta = (ay-cx)/(bx-cy)$, yielding
\bb \label{Stein_1m_sameV}
\mathrm{Stein}_\G(\rho_\G\|\sigma_\G) = 
\frac{a y^2+b x^2-2 c x y}{a b-c^2} = 
(s-t)^\intercal\ V^{-1}\ (s-t)\,, 
\ee
which achieves the theoretical upper bound determined by the classical relative entropy between the Wigner distributions.

Now for $m$-mode states, a factorised homodyne measurement acting independently on each of the modes solves again the problem since
\bb
\sup_{\gamma\geq i \Omega} (s-t)^\intercal (V+\gamma)^{-1}(s-t) &= \sup_{\gamma \geq i \Omega} u^\intercal (D+\gamma)^{-1} u \geq \sum_{i=1}^m \sup_{\gamma_i \geq i \Omega_1} u_i^\intercal (D_i+\gamma_i)^{-1} u_i \\ &=  \sum_{i=1}^m u_i^\intercal D_i^{-1} u_i = (s-t)^\intercal V^{-1}(s-t), 
\ee
where we expressed $V$ in terms of its Williamson decomposition as $V = S D S^\intercal$, defined $u \coloneqq S^{-1}(s-t)$ and picked the factorised ansatz $\gamma = \oplus_i^m \gamma_i$, with each $\gamma_i$ optimising the problem independently in the $i$-th mode.
\end{proof}

\section{Decoupling first and second moments}
We next show that, in the asymptotic Stein setting, the optimisation in Eq.~\eqref{Stein_calculation_0} exhibits a fundamental separation between the contributions of first and second moments. In particular, the first-moment contribution is intrinsically single-letter and admits a simple closed-form expression, so that all nontrivial asymptotic regularisation is confined to the second moments alone.

\begin{thm}[(Decoupling of first and second moments)]\label{prop:momentsepapp} For any pair of Gaussian states $\rho_\G$ and $\sigma_\G$, with first moments $s$ and $t$, and covariance matrices $V$ and $W$, respectively, the corresponding Gaussian Stein exponent factorises as
\bb\label{eq:sfactor}
\mathrm{Stein}_\G(\rho_\G\|\sigma_\G) = \mathrm{Stein}_\G(\rho^0_\G\|\sigma^0_\G) + (s-t)^\intercal W^{-1}(s-t), 
\ee
where $\rho^0_\G$ and $\sigma^0_\G$ are the Gaussian states obtained by setting the first moments of $\rho_\G$ and $\sigma_\G$ to zero.
\end{thm}
    \begin{proof} We start representing the moment space as $(\R^{2m})^{\oplus n} \cong \R^{2m} \otimes \R^n$. The symplectic form then reads $\Omega_m \otimes \id_n$. Now consider the matrix $S_n \coloneqq \id_{2m} \otimes O_n$ with $O_n$ orthogonal. Then $S_n$ is symplectic, therefore given a seed $\gamma$, also $S_n \gamma S_n^T$ is a valid measurement seed. Note also that $S_n V^{\oplus n} S_n^T = V^{\oplus n}$. Hence, re-parameterising with these seed \eqref{Stein_calculation_0} has the only effect of transforming the first moments as $(s-t)^{\oplus n} \to S_n^T (s-t)^{\oplus n}$. Then $(s-t)\otimes u \to (s-t)\otimes O_n^Tu$, where $u$ is the all-one vector. The whole optimisation problem remains unaffected, since it is on every possible seed. In particular, we can choose $O_n$ to satisfy $O_n^T u =\sqrt{n}(1,0,...,0)^T$. This implies that 
    \bb\label{eq:mapping}
    2 \Tr \left[ \left(W^{\oplus n} \!+\! \gamma \right)^{-1} \left( (s\!-\!t)^{\oplus n} \right) \left( (s\!-\!t)^{\oplus n} \right)^\intercal \right] 
    \longmapsto 2 \Tr \left[ \left(W^{\oplus n} \!+\! \gamma \right)^{-1} \left( n (s\!-\!t)(s\!-\!t)^\intercal \oplus 0^{\oplus {n-1}} \right) \right] ,
    \ee
    while the other terms in \eqref{Stein_calculation_0} remain unchanged.
    Let us now rewrite $\gamma$ in the block form 
    \bb
    \gamma = \begin{pmatrix}
    \gamma_A & \gamma_{AB} \\
    \gamma_{AB}^\intercal & \gamma_B
    \end{pmatrix},
    \label{eq:sep1}
    \ee
    where $\gamma_A$ acts on one copy, and $\gamma_B$ acts on the remaining $n-1$ copies. 
    To show that the right-hand side in Eq.~(\ref{eq:sfactor}) provides a lower bound for the left-hand side, we can choose an ansatz for $\gamma$ with $\gamma_{AB} = 0$. Plugging into \eqref{eq:sep1}, this allows $\gamma_A$ to optimise the term involving the first moments (which is additive according to proposition~\ref{thm:firstmoments}) and $\gamma_B$ to optimise the term involving the second moments. Sending $n\to\infty$, the single-copy second-moment term disappears, hence this measurement achieves Eq.~(\ref{eq:sfactor}).
    We can then proceed to establish the converse bound. After the $S_n$ mapping, the Stein exponent can be rewritten as
    \bb\label{eq:sfactor2}
    \lim_{n\to\infty}\frac{1}{n} \sup_{\gamma\geq i\Omega} \frac{1}{2} \biggl[ &\log \frac{\det \left(W^{\oplus n} \!+\! \gamma \right)}{\det \left(V^{\oplus n} \!+\! \gamma\right)} + \Tr (V-W)(W+\Tilde{\gamma}_A^W)^{-1} \\
    &+ \Tr (V-W)^{\oplus(n-1)}(W^{\oplus(n-1)}+\Tilde{\gamma}_B^W)^{-1} + 2n\Tr (s-t)(s-t)^\intercal (W+\Tilde{\gamma}_A^W)^{-1}  \biggl],
    \ee
    where 
    \bb
    \Tilde{\gamma}_A^W \coloneqq (W^{\oplus n} + \gamma)/(W^{\oplus(n-1)}+\gamma_B) - W,
    \ee
    and
    \bb
    \Tilde{\gamma}_B^W \coloneqq (W^{\oplus n} + \gamma)/(W+\gamma_A) - W^{\oplus(n-1)},
    \ee
    are two valid QCMs since $W^{\oplus n} + \gamma \geq W^{\oplus n} + i \Omega$, and the monotonicity of the Schur complement implies $\Tilde{\gamma}_A^W,\Tilde{\gamma}_B^W\geq i \Omega$ by taking the Schur complement on both sides of the inequality with respect to blocks $B$ and $A$ respectively. We now proceed by upper bounding the first term in Eq.~(\ref{eq:sfactor2}) by
    \bb\label{eq:sbound1}
    \log \frac{\det \left(W^{\oplus n} \!+\! \gamma \right)}{\det \left(V^{\oplus n} \!+\! \gamma\right)} &\leqt{(i)} \log \frac{\det \left(W^{\oplus n} \!+\! \gamma \right)}{\det \left({W\oplus V^{\oplus (n-1)}}\!+\! \gamma\right)} + \frac{\|W-V\|_1}{\lambda_{\min}(V)} \\
    &\eqt{(ii)} \log \frac{\det \left(W^{\oplus (n-1)} \!+\! \Tilde{\gamma}_B^W \right)}{\det \left(V^{\oplus (n-1)} \!+\! \Tilde{\gamma}_B^W\right)} + \frac{\|W-V\|_1}{\lambda_{\min}(V)},
    \ee
    where in (i) we used that for any pair of matrices $A,X$ with $A>0$, $X=X^\dagger$, and $A+X>0$,
    \bb
    \frac{\det(A+X)}{\det A} = \det\left( \id + A^{-1/2}XA^{-1/2} \right) = \exp \Tr \ln(\id + A^{-1/2}XA^{-1/2}) \leq \exp \Tr A^{-1}X \leq e^{\|X\|_1/\lambda_{\min}(A)},
    \ee
    and in particular we made the choice $A = V^{\oplus n}+\gamma$ and $X = (W-V){\oplus 0^{\oplus (n-1)}}$ and noted that $\lambda_{\min}(V^{\oplus n}+\gamma) \geq \lambda_{\min}(V^{\oplus n}) = \lambda_{\min}(V)$ 
    because $\gamma\geq 0$. In (ii) instead we used that $\det \left({W\oplus V^{\oplus (n-1)}} \!+\! \gamma\right) = \det \left(V^{\oplus (n-1)} \!+\! \Tilde{\gamma}_B^W\right) \det \left( W + \gamma_A \right) $. We can also simply bound the second term in Eq.~(\ref{eq:sfactor2}) as 
    \bb\label{eq:sbound2}
    \Tr (V-W)(W+\Tilde{\gamma}_A^W)^{-1} \leq \frac{\|W-V\|_1}{\lambda_{\min}\left( W \right)},
    \ee
    and the fourth as
    \bb
    2n\Tr (s-t)(s-t)^\intercal (W+\Tilde{\gamma}_A^W)^{-1} \leq 2n\Tr (s-t)(s-t)^\intercal W^{-1}.
    \ee
    Combining Eqs.~(\ref{eq:sbound1}) and (\ref{eq:sbound2}) we can upper bound the term in the square brackets in Eq.~(\ref{eq:sfactor2}) by
    \bb
    \log \frac{\det \left(W^{\oplus (n-1)} \!+\! \Tilde{\gamma}_B^W \right)}{\det \left(V^{\oplus (n-1)} \!+\! \Tilde{\gamma}_B^W\right)} + \|W-V\|_1 \left( \frac{1}{\lambda_{\min}(V)} + \frac{1}{\lambda_{\min}\left( W \right)}\right) + \Tr (V-W)^{\oplus(n-1)}(W^{\oplus(n-1)}+\Tilde{\gamma}_B^W)^{-1} \\
    +  2n\Tr (s-t)(s-t)^\intercal W^{-1}.
    \ee
    Taking the supremum, dividing by $n$, sending $n \to \infty$, and noting that the remainder terms in Eqs.~(\ref{eq:sbound1}) and (\ref{eq:sbound2}) do not depend on $\gamma$ or $n$ and thus vanish in the many-copy limit, we establish the original claim.
    \end{proof}

\section{The problem of additivity}
We have established in the previous section that the Gaussian Stein exponent splits in the asymptotic limit into a single-letter term for the first moments, and a separate optimisation of the remaining $n-1$ copies in the purely second moment contribution. The term $\mathrm{Stein}_\G(\rho^0_\G\|\sigma^0_\G)$ remains a regularised quantity, and is thus difficult to calculate. In this section we focus purely on the second moment optimisations in \eqref{Stein_calculation_0}, i.e. the contribution
\bb
\mathrm{Stein}_\G(\rho^0_\G\|\sigma^0_\G) = \lim_{n\to\infty} \frac1n \sup_{\gamma \geq i\Omega^{\oplus n}} \Bigg\{ \frac12 \log \frac{\det \left(W^{\oplus n} \!+\! \gamma \right)}{\det \left(V^{\oplus n} \!+\! \gamma\right)}+ \frac12 \Tr \left[ \left(W^{\oplus n} \!+\! \gamma \right)^{-1} (V\!-\!W)^{\oplus n} \right] \Bigg\}.
\label{eq:steincov}
\ee

To this end, we begin by proving a lemma relating log-determinants of sums of matrices to the overlaps of their eigenvectors.
{\begin{lemma}\label{lemma:logdetoverlap}
Let \(A,B>0\) be real symmetric matrices with eigendecomposition
\bb
A=\sum_i\alpha_i \ketbra{\psi_i} \qquad \text{and}\qquad
B=\sum_j\beta_j\ketbra{\phi_j}.
\ee
Where \(\{\psi_i\}\) and \(\{\phi_j\}\) are real orthonormal
eigenbases.
Then, the following inequality holds:
\bb\label{eq:op_ineq}
\log\det(A+B)
\ge
\sum_{i,j}
|\braket{\psi_i|\phi_j}|^2
\log(\alpha_i+\beta_j).
\ee
where, if \(A\) and \(B\) commute, equality holds. 
\end{lemma}}

\begin{proof}
{
Consider Hilbert--Schmidt matrix space and define
\bb
L_A(X)=AX,\qquad
R_B(X)=XB.
\ee
Put $M=L_A+R_B$ and denote by $\pazocal{T}$ the transpose map, namely \(\pazocal T(X)=X^T\). Then,
$\pazocal T M\pazocal T=L_B+R_A$.
By the operator concavity of the logarithm, we have
\bb
\log\frac{M+\pazocal T M\pazocal T}{2}
\ge
\frac12
\left(
\log M+\pazocal T(\log M)\pazocal T
\right).
\ee
Take the expectation against \(\id\). Since \(\pazocal T(\id)=\id\), the right-hand side equals
$\langle \id,\log M\,\id\rangle$.
Now,
\bb
\frac{M+\pazocal T M\pazocal T}{2}=\frac12(L_{A+B}+R_{A+B}).
\ee

If \(c_k\) are the eigenvalues of \(A+B\), then the expectation against \(\id\) of the logarithm of this operator is
\bb
\sum_k\log c_k
=
\log\det(A+B).
\ee
On the other hand, \(\psi_i\phi_j^T\)
is an eigenvector of \(M\) with eigenvalue \(\alpha_i+\beta_j\),
and
\[
\langle \psi_i\phi_j^T,\id\rangle
=\langle \psi_i,\phi_j\rangle.
\]
Therefore
\bb
\langle \id,\log M\,\id\rangle
=
\sum_{i,j}
|\langle \psi_i,\phi_j\rangle|^2
\log(\alpha_i+\beta_j),
\ee
proving the claim.}
\end{proof}

We are now in a position to prove the main result of this section. We show here sufficient conditions under which the covariance optimisation \eqref{eq:steincov} becomes additive, and we can therefore remove the regularisation.
\begin{thm}[(Sufficient conditions for additivity)]\label{thm:sufficient} Let \(V,W\) be \(m\)-mode quantum covariance matrices. Assume that \(W\) has flat symplectic spectrum, namely, there is a symplectic matrix \(S\) such that $SWS^\intercal = w\id_{2m}$.
If, in addition, 
\bb\label{eq:condition}
[\widetilde V,\Omega\widetilde V\Omega^\intercal]=0,
\ee
where  $\widetilde V\coloneqq SVS^\intercal$, then we have
\bb
    \rel{D^{\G}}
{V^{\oplus n}}{W^{\oplus n}}
=
n\rel{D^{\G}}{V}{W},
\ee
which immediately implies $\rel{D^{\G,\infty}}{V}{W}=\rel{D^{\G}}{V}{W}$. 

\end{thm}

\begin{proof}[Proof of Theorem \ref{thm:sufficient}]
    The proof is divided in four steps.\smallskip
    
    \textbf{Step 1.} First, we show that \eqref{eq:condition} does not depend on the specific choice of the symplectic matrix diagonalising $W$, but only on the pair $(V,W)$. Indeed, suppose 
    \bb
    S_1WS_1^\intercal=S_2WS_2^\intercal=w\id_{2m},
    \ee
    and set $O\coloneqq S_2S_1^{-1}$. Then, $O$ is orthogonal and symplectic, namely $O\in {\rm Sp}(2m,\R)\cap {\rm O}(2m)$. Moreover, calling $\widetilde V_i\coloneqq S_iVS_i^\intercal$ for $i=1,2$, we immediately have $\widetilde V_2=O\widetilde V_1O^\intercal$. Since any orthogonal symplectic matrix commutes with $\Omega$, we get
\bb
\Omega\widetilde V_2\Omega^\intercal
= O(\Omega\widetilde V_1\Omega^\intercal)O^\intercal,
\ee
whence
\bb
[\widetilde V_2,\Omega\widetilde V_2\Omega^\intercal]
=O[\widetilde V_1,\Omega\widetilde V_1\Omega^\intercal]O^\intercal,
\ee
which shows our claim.\smallskip

\textbf{Step 2.} Now we want to prove that, for any arbitrary real and symmetric matrix $\widetilde V$, we have
\bb\label{eq:claim1}
[\widetilde V,\Omega \widetilde V\Omega^\intercal]=0\qquad \iff \qquad \exists O\in {\rm O}(2m)\cap{\rm Sp}(2m,\R) : O\widetilde VO^\intercal
=
\diag
(a_1,\ldots,a_m,b_1,\ldots,b_m).
\ee
Consider the decomposition
\bb\label{eq:decomposition}
\widetilde V=\widetilde V_++\widetilde V_-\qquad\text{with}\qquad \widetilde V_+\coloneqq \frac12(\widetilde V+\Omega \widetilde V\Omega^\intercal),
\qquad
\widetilde V_-\coloneqq \frac12(\widetilde V-\Omega \widetilde V\Omega^\intercal).
\ee
Note that $[\widetilde V_+,\Omega]=0$ and $\{\widetilde V_-,\Omega\}=0$. The above decomposition clearly satisfies $\Omega \widetilde V\Omega^\intercal=\widetilde V_+-\widetilde V_-$, therefore
\bb
[\widetilde V,\Omega \widetilde V\Omega^\intercal]
=
[\widetilde V_++\widetilde V_-,\widetilde V_+-\widetilde V_-]
=
-2[\widetilde V_+,\widetilde V_-].
\ee
The left-hand side of \eqref{eq:claim1} is thus equivalent to
\bb\label{eq:assumption_commute}
[\widetilde V_+,\widetilde V_-]=0.
\ee
Consider the diagonalisation \(\widetilde V_+=\sum_i\lambda_i\Pi_i\) of the real symmetric matrix. Since \(\widetilde V_+\) commutes with \(\Omega\), every eigenspace of \(\widetilde V_+\) is invariant under \(\Omega\).
Furthermore, under the assumption \eqref{eq:assumption_commute} that \(\widetilde V_-\) commutes with \(\widetilde V_+\), each eigenspace is also invariant under \(\widetilde V_-\).
Hence, by working separately on every subspace $\pazocal{V}_i=\ran \Pi_i$, we can proceed without loss of generality under the assumption that $\widetilde V_+=c\id$ for some $c\in \mathbb{R}$. Since \(\widetilde V_-\) is real symmetric and anticommutes with \(\Omega\), it is known that the eigenvalues of \(\widetilde V_-\) occur in pairs \(s,-s\), with eigenvectors related by \(\Omega\)\footnote{Indeed, suppose $\widetilde V_-u=su$. Then, $\widetilde V_-(\Omega u)=-\Omega \widetilde V_-u=-s\,\Omega u$.}.
Within \(\pazocal V_i\), choose an orthonormal eigenbasis
\(\{q_j\}\) of the strictly positive eigenspaces of
\(\widetilde V_-\), and define \(p_j\coloneqq-\Omega q_j\).
These partners form an orthonormal basis of the strictly
negative eigenspaces.
On the kernel, construct the pairs iteratively: choose a
unit vector \(q_j\), set \(p_j=-\Omega q_j\), and repeat in
the orthogonal complement of the pairs already chosen,
assigning \(s_j=0\) to these pairs.
Since \(\Omega\) is orthogonal and skew-symmetric, each pair
is orthonormal; moreover, the kernel and each remaining
orthogonal complement are invariant under \(\Omega\),
so the construction exhausts the kernel.
Then,
\bb
\Omega q_j=-p_j,\qquad
\Omega p_j=q_j,\qquad
 \widetilde V_-q_j=s_jq_j,
\qquad
\widetilde V_-p_j=-s_jp_j.
\ee
The vectors \(\{q_j,p_j\}\) form an orthonormal symplectic basis. In that basis,
\bb
\widetilde V_+ = c\id \qquad \widetilde V_-=
\diag
(s_1,\ldots,s_d,-s_1,\ldots,-s_d),
\ee
therefore $\widetilde V$ has the form of the right-hand side of \eqref{eq:claim1}:
\bb
\widetilde V=\widetilde V_++\widetilde V_-=\diag
(c+s_1,\ldots,c+s_d,
c-s_1,\ldots,c-s_d)
\ee
Doing this independently in every eigenspace of \(\widetilde V_+\) produces an orthogonal symplectic \(O\) completing the direct implication of \eqref{eq:claim1}. Conversely, if \(\widetilde V\) already has the form as in the right-hand side of \eqref{eq:claim1}, then
\[
\Omega \widetilde V\Omega^\intercal
=
\diag
(b_1,\ldots,b_m,a_1,\ldots,a_m),
\]
which clearly commutes with \(\widetilde V\). Orthogonal-symplectic conjugation preserves this commutator.\smallskip

\textbf{Step 3.} Now, under \eqref{eq:condition}, let us apply the direct implication of the previous step: there exists $O\in {\rm O}(2m)\cap{\rm Sp}(2m,\R)$ such that
\bb
O\widetilde VO^\intercal
=
\begin{pmatrix}
A&0\\
0&B
\end{pmatrix}\qquad\text{with}\qquad A=\diag(a_1,\ldots,a_m),
\qquad
B=\diag(b_1,\ldots,b_m).
\ee
Since \(O\) is orthogonal, the common symplectic transformation $OS$
puts the pair $(V,W)$ into the pair
\bb\label{eq:form}
V'=(OS)V(OS)^\intercal=
\diag
(a_1,\ldots,a_m,b_1,\ldots,b_m),
\qquad W'=w\id_{2m}.
\ee
The Gaussian measured relative entropy is invariant under a common symplectic transformation: any Gaussian measurement seed is transformed bijectively to another Gaussian measurement seed, and the corresponding classical outcomes differ only by an invertible linear coordinate change. Hence, without loss of generality, it is sufficient to prove Theorem \ref{thm:sufficient} in the particular case of covariance matrices of the form \eqref{eq:form}.\smallskip

When considering \(n\) copies of the bosonic system, by the previous step we can put $V^{\oplus n}$ and $W^{\oplus n}$ in the form $W=w\id_{2nm}$ and
\bb\label{eq:ntimes}
    V_n&=\diag(a_1,\ldots,a_m,a_{m+1}, \ldots,a_N, b_1,\ldots,b_m,b_{m+1}\ldots, b_N)\\
    & \coloneqq
\diag
(\underbrace{a_1,\ldots,a_m,\ldots, a_1,\ldots,a_m}_{n \text{ times}}, \underbrace{b_1,\ldots,b_m,\ldots, b_1,\ldots,b_m}_{n \text{ times}})
\ee
where, for convenience, we have identified $a_r$ with $a_{1+(r-1\!\!\mod\! m)}$.
Recall that, from \eqref{eq:steincov}, for a given Gaussian seed \(\Gamma\), the centred classical Gaussian relative entropy $F_n(\Gamma)$ is
\bb\label{eq:rel_ent}
F_n(\Gamma)=\frac12
\left[
\Tr
(W_n+\Gamma)^{-1}(V_n+\Gamma)-2N
-\log\det(V_n+\Gamma)
+\log\det(W_n+\Gamma)
\right], \qquad N=nm.
\ee
Taking the supremum over physical Gaussian seeds, we get $D^{\mathrm G}(V_n\|W_n)$. We want to show that in such optimisation pure seeds are sufficient. Given any arbitrary seed $\Gamma$, with Williamson decomposition $\Gamma=TDT^\intercal$.
The covariance matrix $\Gamma_0\coloneqq TT^\intercal$ is pure, satisfies
\bb
\Gamma-\Gamma_0\ge0,
\ee
and yields a larger relative entropy, namely,
\bb\label{eq:dpi}
    F_n(\Gamma)\le F_n(\Gamma_0).
\ee
Indeed, the output covariance associated with \(\Gamma\) is obtained from that associated with \(\Gamma_0\) by adding the same independent Gaussian classical noise under both hypotheses.
Therefore classical data processing gives \eqref{eq:dpi}.
Thus we may restrict the optimisation to pure seeds.

Every pure \(N\)-mode QCM can be written in the  Bloch-Messiah form \cite{Arvind_1995,PhysRevA.71.055801}
\bb\label{eq:form1}
\Gamma=K
\begin{pmatrix}
X&0\\
0&X^{-1}
\end{pmatrix}K^\intercal, \qquad \text{where}\qquad X=\diag(x_1,\ldots,x_N),
\qquad x_\ell>0,\quad \text{and}\quad K\in {\rm Sp}(2N,\R)\cap{\rm O}(2N),
\ee
where $N=nm$.
Furthermore, in grouped coordinates $(Q_1,\ldots,Q_N,P_1,\ldots,P_N)$, every orthogonal-symplectic matrix $K$ has the form
\bb\label{eq:K}
K=
\begin{pmatrix}
C&-S\\
S&C
\end{pmatrix},
\qquad
C+iS\in U(N).
\ee
In particular, for every $1\leq \beta\leq N$, from orthogonality of $K$ we have
\bb\label{eq:1}
\sum_{\alpha=1}^{2N} (K_{\alpha\beta})^2=\sum_{\alpha=1}^N
((C_{\alpha\beta})^2+(S_{\alpha\beta})^2)=1.
\ee
Let
$\mathfrak x=
(x_1,\ldots,x_N,x_1^{-1},\ldots,x_N^{-1})$ and 
$v=
(a_1,\ldots,a_N,b_1,\ldots,b_N)$
be the diagonal entries $\text{diag}(X,X^{-1})$ and
 \(V_n\), respectively.
The operator inequality in lemma \ref{lemma:logdetoverlap}, applied to \(V_n\) and \(\Gamma\), yields
\bb\label{eq:18}
\log\det(V_n+\Gamma)
\ge
\sum_{\alpha,\beta}
K_{\alpha\beta}^2
\log(v_\alpha+\mathfrak x_\beta).
\ee
while, for \(W_n=w\id\), we can write
\bb\label{eq:19}
\log\det(w\id+\Gamma)
=
\sum_{\alpha,\beta}
K_{\alpha\beta}^2
\log(w+\mathfrak x_\beta).
\ee
using the fact that $\sum_\alpha K_{\alpha\beta}^2=1$.
Similarly, we can write
\bb\label{eq:20}
\Tr
[(w\id+\Gamma)^{-1}(V_n+\Gamma)]
=
\sum_{\alpha,\beta}
K_{\alpha\beta}^2
\frac{v_\alpha+\mathfrak x_\beta}{w+\mathfrak x_\beta}.
\ee
Using the inequalities and identities \eqref{eq:18}, \eqref{eq:19} and \eqref{eq:20} in \eqref{eq:rel_ent}, we can convert arbitrary collective passive mode mixing into classical doubly-stochastic weights:
\bb\label{eq:upper}
F_n(\Gamma)
\le
\frac12
\sum_{\alpha,\beta}
K_{\alpha\beta}^2
f\!\left(
\frac{v_\alpha+\mathfrak x_\beta}{w+\mathfrak x_\beta}
\right)=\frac 12\sum_{r,\ell=1}^N \left\{C_{r\ell}^2 \left[f\!\left(
\frac{a_r+x_\ell}{w+x_\ell}\right)+f\!\left(
\frac{b_rx_\ell+1}{wx_\ell+1}\right)\right]+S_{r\ell}^2 \left[f\!\left(
\frac{b_r+x_\ell}{w+x_\ell}\right)+f\!\left(
\frac{a_rx_\ell+1}{wx_\ell+1}\right)\right]\right\}
\ee
where $f(y)=y-1-\log y$, and the second identity is obtained using the block structure \eqref{eq:K} and making explicit the coordinates of the vectors $\mathfrak x$ and $v$. Now, calling  
\bb
    \mathfrak F_r \coloneqq \sup_{x>0} \frac 12\left[f\!\left(
\frac{a_r+x}{w+x}\right)+f\!\left(
\frac{b_rx+1}{wx+1}\right)\right]
\ee
and recalling \eqref{eq:1}, the upper bound \eqref{eq:upper} can be relaxed to
\bb
    F_n(\Gamma) \le \sum_{r=1}^N \mathfrak F_r\sum_{\ell=1}^N(C_{r\ell}^2+S_{r\ell}^2)
    &=\sum_{r=1}^N \mathfrak F_r=n\sum_{j=1}^m \mathfrak F_j.
\ee
where in the last identity we have used that, according to \eqref{eq:ntimes}, each physical mode \(j\) occurs \(n\) times.
Since we have a global upper bound on the family of collective pure seeds, which are sufficient in the optimisation for $D^{\mathrm G}(V_n\|W_n)$, we get the upper bound
\bb
D^{\mathrm G}(V_n\|W_n)\le n\sum_{j=1}^m \mathfrak F_j.
\ee

\textbf{Step 4.} Finally, we exhibit that the previous upper bound can be attained by product seeds. 
For each physical mode \(j\), choose \(x_j>0\) in order to construct the one-copy product seed $\Gamma_*$ in grouped coordinated $(Q_1,P_1,\ldots,Q_m,P_m)$ as follows
\bb
 \Gamma_j\coloneqq
 \begin{pmatrix}
x_j&0\\
 0&x_j^{-1}
\end{pmatrix}, 
\qquad \Gamma_*\coloneqq
\bigoplus_{j=1}^m\Gamma_j.
\ee
On \(n\) copies we then consider $\Gamma_*^{\oplus n}$.
Since all the quantities are diagonal and commute, \eqref{eq:op_ineq} is saturated. In particular, optimising over $x_j>0$, we get
\bb
D^{\mathrm G}(V_n\|W_n)
\geq \sup_{\substack{x_j>0\\ 1\leq j\leq m}}F_n(\Gamma_*^{\oplus n})=n\sum_{j=1}^m \mathfrak F_j.
\ee
{In the case where the optimal $F_n(\Gamma_*^{\oplus n})$ is attained in the homodyne-like limit, corresponding to one of the limits $x_j\to0$ or $x_j\to\infty$ on some number of modes, we take the appropriate finite-squeezing limit on those modes.}
Since the upper and the lower bounds coincide, we conclude that
\bb
D^{\mathrm G}(V_n\|W_n) =n\sum_{j=1}^m
\mathfrak F_j=nD^{\mathrm G}(V\|W)
\ee
which completes the proof of the main claim.\smallskip

\end{proof}

Using the results of Theorem~\ref{thm:sufficient}, dividing by $n$ and taking the limit $n\to\infty$ in \eqref{eq:steincov}, we arrive at
\bb
\mathrm{Stein}_\G(\rho^0_\G\|\sigma^0_\G)= \sup_{\gamma \geq i\Omega} \Bigg\{ \frac12 \log \frac{\det \left(W \!+\! \gamma \right)}{\det \left(V \!+\! \gamma\right)}+ \frac12 \Tr \left[ \left(W \!+\! \gamma \right)^{-1} (V\!-\!W) \right] \Bigg\}=D^G(V\|W).
\ee
Combining this with the splitting of first and second moments in Theorem~\ref{thm:momentsep}, we arrive at
\bb
\mathrm{Stein}_\G(\rho_\G\|\sigma_\G)=D^G(V\|W)+(s-t)^\top W^{-1}(s-t).
\label{eq:gsteinsimple}
\ee

Theorem~\ref{thm:sufficient} provides sufficient conditions for the above formulation of the Gaussian Stein exponent \eqref{eq:gsteinsimple}. These conditions apply, for example, to all single-mode systems and all pure multimode systems. In the single-mode case, \(W\) automatically has a single symplectic eigenvalue $w=\sqrt{\det W}$, and for every positive symmetric \(2\times2\) matrix \(A\), we have $\Omega A\Omega^\intercal=(\det A)A^{-1}$, whence $[A,\Omega A\Omega^\intercal]=0$, thus every one-mode pair satisfies the Theorem. For a pure multimode covariance matrix $W$, by definition we have \(w=1\). By a symplectic transformation, we can map \(W\to \id\), while leaving \(\widetilde V\) pure, so that $\widetilde V\Omega\widetilde V=\Omega$.
Hence
\bb
\Omega\widetilde V\Omega^T
=
\widetilde V^{-1},\qquad \text{i.e.}\qquad [\widetilde V,\Omega\widetilde V\Omega^T]=[\widetilde V,\widetilde V^{-1}]
=0.
\ee
Thus every pure multimode pair also satisfies the condition stated in Theorem~\ref{thm:sufficient}.

\subsection{Single-mode example}\label{sec:singlemodeapp}
In this section we consider the Gaussian discrimination problem between two single-mode Gaussian states. According to Theorem~\ref{thm:momentsep}, we can consider separate optimisations over the terms involving first and second moments. The first moments term is already solved by a suitably rotated homodyne measurement, according to proposition~\ref{thm:firstmoments}. The term involving only second moments reduces to a single-copy optimisation, according to Theorem~\ref{thm:sufficient}. We are thus left to solve this single-copy second moment optimisation.

By exploiting symplectic invariance, we can parameterise an arbitrary pair of centred, single-mode Gaussian states via
\bb
V =  a  \begin{pmatrix} \mu & 0 \\ 0 & 1/\mu \end{pmatrix} \quad\text{and}\quad W =  b  \id_2,
\ee
where $a,b\geq1$, $\mu>0$,
and an arbitrary measurement seed via
\bb
\gamma=R_\theta \begin{pmatrix}
    z&0\\0&\frac1z
\end{pmatrix}
R_\theta^\top
\quad \text{with}\quad
R_\theta=\begin{pmatrix}
    \cos\theta &-\sin\theta\\
    \sin\theta&\cos\theta 
\end{pmatrix}.
\ee
Depending on the interplay between the state parameters $\mu,  a, b$, the Gaussian Stein exponent is achieved by a single-mode seed with no rotation ($\theta=0$), amounting either to a homodyne detection on the $x$- or $p$-axis ($z\rightarrow 0$ or $z \rightarrow \infty$), or a general-dyne measurement corresponding to projection onto a finitely squeezed Gaussian state with $z=z_{\text{opt}}$. The resulting Gaussian Stein exponent is then
\bb
\begin{aligned}
\mathrm{Stein}_\G(\rho_G\|\sigma_G) = \frac12 \max\bigg\{&\frac{\mu a }{ b } + \log\left(\frac{ b }{\mu a }\right) - 1, \frac{ a }{ b \mu} + \log\left(\frac{ b \mu}{ a }\right)  - 1,  \\
&
\frac{a \mu -b}{b+z_{\text{opt}}}+\frac{z_{\text{opt}} (a-b \mu )}{b \mu  z_{\text{opt}}+\mu }+\log \left(\frac{\mu  \left(b+z_{\text{opt}}\right) \left(b z_{\text{opt}}+1\right)}{\left(a \mu +z_{\text{opt}}\right) \left(a z_{\text{opt}}+\mu \right)}\right)\bigg\},
\end{aligned}
\label{eq:steinsinglemodesolution}
\ee
with $z_{\text{opt}}$ chosen as one of the (at most four) roots of 
\[
P(z)=\frac{ a  - b   \mu }{\mu  ( b   z+1)^2}+\frac{ b  -\mu   a  }{( b  +z)^2}+\frac{1}{ b  +z}-\frac{1}{z(1+bz)}+\frac{\mu }{z (\mu + a   z)}-\frac{1}{\mu   a  +z},
\]
such that the third term in \eqref{eq:steinsinglemodesolution} is maximised.
The first two terms in the maximisation correspond to homodyne measurements of the x- or p-axis. The third term corresponds to a general-dyne measurement, which for $z_{\text{opt}}=1$ constitutes heterodyne detection. There exist examples when this general-dyne measurement is the optimal choice, for example $\mu=6,  a =3, b =5$; this gives $z_{\text{opt}} \approx 1.9$.


\begin{figure} 
\includegraphics[width=12cm]{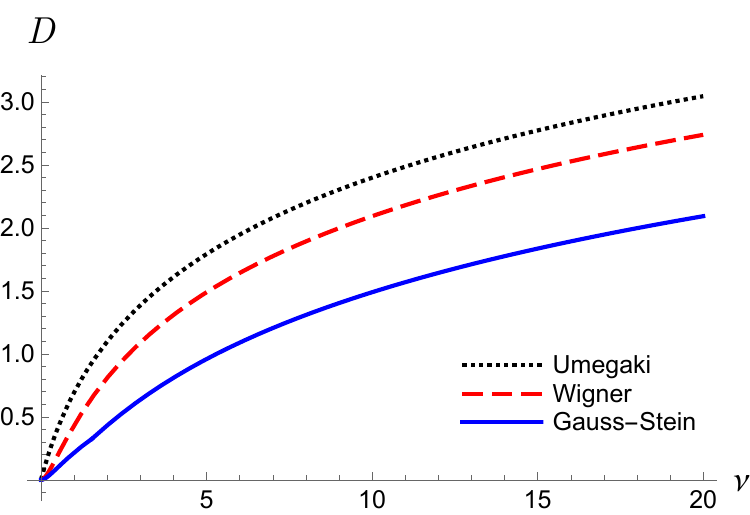}
\caption{Discriminating between a single-mode vacuum state and a single-mode thermal state with mean photon number $\nu$. The black dotted line is the Umegaki relative entropy corresponding to the unrestricted Stein exponent. The dashed red curve is the classical Wigner relative entropy upper bound. The blue line is the Gaussian Stein exponent. \label{steinthermal_fig}}
\end{figure}

For ease of illustration, let $V$ and $W$ be respectively the vacuum state ($a=1$, $\mu=1$) and a thermal state ($b=2\nu+1$) with mean photon number $\nu \geq 0$. In this case, the upper bound for the Stein exponent determined by the classical relative entropy between Wigner distributions is
\bb
D(V/2,0 \| W/2,0) = \log(2\nu+1)+(2\nu+1)^{-1}-1\,,
\ee
while the Umegaki relative entropy \cite{Umegaki1962} equals $\log(\nu+1)$. The Gaussian Stein exponent from Eq.~(\ref{eq:steinsinglemodesolution}) is given by
\bb 
\mathrm{Stein}_\G(\rho_G\|\sigma_G) = \max \left\{\frac12\left( \log(2\nu+1)+(2\nu+1)^{-1}-1 \right), \log(\nu+1)+(\nu+1)^{-1}-1\right\}\,.
\ee
The Umegaki relative entropy, classical relative entropy between Wigner distributions, and Gaussian Stein exponent for these states are plotted in Figure~\ref{steinthermal_fig}.

\end{document}